\documentclass[runningheads]{llncs}

\newcommand{\occ}[2]{\mathrm{occ} (#1, #2)}
\newcommand{\pred}{\mathrm{pred} }
\newcommand{\last}{\mathrm{last} }
\newcommand{\eps}{\varepsilon}
\newcommand{\tp}{\tilde{p}}
\newcommand{\calC}{{\mathcal C}}
\newcommand{\cN}{{\mathcal N}}
\newcommand{\cost}[1]{\ensuremath{\mathrm{cost} (#1)}}

\newenvironment{enumerate*}%
  {\vspace{-1ex}\begin{enumerate}%
    \setlength{\leftmargin}{1pt}%
    \setlength{\itemindent}{-3ex}
    \setlength{\parsep}{0pt}
  }%
  {\end{enumerate}}

\begin{document}

\title{Worst-Case Optimal Adaptive Prefix Coding}
\author{Travis Gagie\inst{1}\fnmsep\thanks
    {This paper was written while the second author was at the University of Eastern Piedmont, Italy, supported by Italy-Israel FIRB Project ``Pattern Discovery Algorithms in Discrete Structures, with Applications to Bioinformatics''.} \and Yakov Nekrich\inst{2}}
\authorrunning{T. Gagie and Y. Nekrich}
\institute{Research Group in Genome Informatics\\
    University of Bielefeld, Germany\\
    \email{travis.gagie@gmail.com}\\ \mbox{} \\
    Department of Computer Science\\
    University of Bonn, Germany\\
    \email{yasha@cs.uni-bonn.de}}
\maketitle

\begin{abstract}
  A common complaint about adaptive prefix coding is that it is much
  slower than static prefix coding.  Karpinski and Nekrich recently
  took an important step towards resolving this: they gave an adaptive
  Shannon coding algorithm that encodes each character in \(O (1)\)
  amortized time and decodes it in \(O (\log H)\) amortized time,
  where $H$ is the empirical entropy of the input string $s$.  For
  comparison, Gagie's adaptive Shannon coder and both Knuth's and
  Vitter's adaptive Huffman coders all use \(\Theta (H)\) amortized
  time for each character.  In this paper we give an adaptive Shannon
  coder that both encodes and decodes each character in \(O (1)\)
  worst-case time.  As with both previous adaptive Shannon coders, we
  store $s$ in at most \((H + 1) |s| + o (|s|)\) bits. We also show
  that this encoding length is worst-case optimal up to the lower order term.
\end{abstract}

\section{Introduction} \label{sec:introduction}

Adaptive prefix coding is a well studied problem whose well known and
widely used solution, adaptive Huffman coding, is nevertheless not worst-case optimal.  Suppose we are given a string $s$ of length
$m$ over an alphabet of size $n$.  For
static prefix coding, we are allowed to make two passes over $s$ but,
after the first pass, we must build a single prefix code, such as a
Shannon code~\cite{Sha48} or Huffman code~\cite{Huf52}, and use it to
encode every character.  Since a Huffman code minimizes the expected
codeword length, static Huffman coding is optimal (ignoring the
asymptotically negligible \(O (n \log n)\) bits needed to write the
code).  For adaptive prefix coding, we are allowed only one pass over
$s$ and must encode each character with a prefix code before reading
the next one, but we can change the code after each character.
Assuming we compute each code deterministically from the prefix of $s$
already encoded, we can later decode $s$ symmetrically.  The most
intuitive solution is to encode each character using a Huffman code
for the prefix already encoded;  Knuth~\cite{Knu85} showed how to do
this in time proportional to the length of the encoding produced,
taking advantage of a property of Huffman codes discovered by
Faller~\cite{Fal73} and Gallager~\cite{Gal78}.  Shortly thereafter,
Vitter~\cite{Vit87} gave another adaptive Huffman coder that also uses
time proportional to the encoding's length; he proved his coder stores
$s$ in fewer than $m$ more bits than static Huffman coding, and that
this is optimal for any adaptive Huffman coder.  With a similar
analysis, Milidi\'u, Laber and Pessoa~\cite{MLP99} later proved
Knuth's coder uses fewer than \(2 m\) more bits than static Huffman
coding.  In other words, Knuth's and Vitter's coders store $s$ in at
most \((H + 2 + h) m + o (m)\) and \((H + 1 + h) m + o (m)\) bits,
respectively, where \(H = \sum_a (\occ{a}{s} / m)
\log (m / \occ{a}{s})\) is the empirical entropy of $s$ (i.e., the entropy of the normalized distribution of characters in $s$), \(\occ{a}{s}\) is the number of
occurrences of the character $a$ in $s$, and \(h \in [0, 1)\) is the redundancy of a Huffman code for $s$; therefore, both adaptive Huffman coders use \(\Theta
(H)\) amortized time to encode and decode each character of $s$.
 Turpin
 and Moffat~\cite{TM01} gave an adaptive prefix coder that uses
 canonical codes, and showed it achieves nearly the same compression as
 adaptive Huffman coding but runs much faster in practice.  Their upper
 bound was still \(O (H)\) amortized time for each character but their
 work raised the question of asymptotically faster adaptive prefix
 coding.
In all of the above algorithms the encoding and decoding times
are proportional to the \emph{bit length} of the encoding. This implies that
we need $O(H)$ time to encode/decode each symbol; since entropy $H$ depends
on the size of the alphabet, the running times grow with
the alphabet size.

The above results for adaptive prefix coding
 are in contrast to the algorithms for the prefix coding in the
static scenario.
The simplest static Huffman coders use \(\Theta (H)\) amortized time
to encode and decode each character but, with a lookup table storing
the codewords, it is not hard to speed up encoding to take \(O (1)\)
worst-case time for each character.
We can also decode an arbitrary prefix code in \(O(1)\) time
using a look-up table, but the space usage and initialization time
for such a table can be prohibitively high, up to $O(m)$.
Moffat and Turpin~\cite{MT97}
described a practical algorithm for decoding prefix
codes in $O(1)$ time; their algorithm works for a special class
of prefix codes, the \emph{canonical codes} introduced by
 Schwartz and Kallick~\cite{SK64}.

While all adaptive coding methods described above maintain the optimal
Huffman code, Gagie~\cite{Gag07} described an adaptive prefix coder that
is based on sub-optimal Shannon coding; his method also
 needs $O(H)$ amortized
time per character for both encoding and decoding.
Although the algorithm  of~\cite{Gag07} maintains
a Shannon code that is known to be worse than the Huffman code
in the static scenario, it achieves $(H+1)m+o(m)$ upper bound on the
encoding length that is better than the best known
upper bounds for adaptive Huffman
algorithms. Karpinski and Nekrich~\cite{KN??} recently
reduced the gap between static and adaptive prefix coding
by using quantized canonical coding to speed up an adaptive Shannon coder
of
Gagie~\cite{Gag07}:  their coder  uses \(O (1)\) amortized time to encode
each character and \(O (\log H)\) amortized time to decode it; the encoding
length  is  also at most \((H + 1) m + o (m)\) bits.

In this paper we describe an algorithm  that both encodes
and decodes each character in \(O (1)\) worst-case time, while still
using at most \((H + 1) m + o (m)\) bits.
It can be shown that the encoding length of
any adaptive prefix coding algorithm
is $(H+1)m -o(m)$ bits in the worst case.
Thus our algorithm works in optimal worst-case time independently of the
 alphabet size  and achieves optimal encoding length (up to the lower-order
term).
 As is common, we assume \(n
\ll m\)\footnote{In fact, our main result is valid if $n=o(m/\log^{5/2}m)$.
For the result of section~\ref{sec:sorting} and two results in
section~\ref{sec:other} we need a somewhat stronger assumption that
$n=o(\sqrt{m/\log m})$}; for simplicity, we also assume $s$ contains at least two
distinct characters and $m$ is given in advance.  Our model is a
unit-cost word RAM with \(\Omega (\log m)\)-bit words on which it
takes \(O (1)\) time to input or output a word.
Our encoding algorithm uses only  addition and bit operations;
our decoding algorithm also uses multiplication and finding
the most significant bit in $O(1)$ time.
We can also implement the decoding algorithm,
so that it uses $AC^0$ operations only.
Encoding needs $O(n)$ words of space, and decoding needs $O(n\log m)$
words of space.
The decoding algorithm can be implemented with bit operations only at a
cost of higher space usage and additional pre-processing time.
For an arbitrary constant $\alpha>0$, we can construct in $O(m^{\alpha})$
time a look-up table that uses $O(m^{\alpha})$ space; this look-up
table enables us to implement multiplications with $O(1)$
table look-ups and bit operations.

While the algorithm of~\cite{KN??} uses quantized coding,
i.e., coding based on the quantized symbol probabilities,
our algorithm is based on delayed probabilities:
encoding of a symbol $s[i]$ uses a Shannon code for the prefix $s[1..i-d]$
for an appropriately chosen parameter $d$; henceforth $s[i]$ denotes the
$i$-th symbol in the string $s$ and $s[i..j]$ denotes the substring
of $s$ that consists of symbols $s[i]s[i+1]\ldots s[j]$.
In
Section~\ref{sec:canonical} we describe canonical Shannon codes and
explain how they can be used to speed up  Shannon coding.  In
Section~\ref{sec:data structures} we describe two useful data
structures that allow us to
maintain the Shannon code efficiently.
 We present our algorithm and analyze the number of bits
it needs to encode a string  in Section~\ref{sec:algorithm}.
In Section~\ref{sec:lower bound} we prove a matching lower bound
by extending Vitter's lower bound from adaptive Huffman coders to all
adaptive prefix coders. In section~\ref{sec:sorting} we show that our technique
can be applied to the online stable sorting problem.
In section~\ref{sec:other} we describe how we can use the same approach
of delayed adaptive coding to achieve $O(1)$ worst-case encoding time
for several other coding problems in the adaptive scenario.

\section{Canonical Shannon coding} \label{sec:canonical}

Shannon~\cite{Sha48} defined the entropy \(H (P)\) of a probability
distribution \(P = p_1, \ldots, p_n\) to be \(\sum_{i = 1}^n p_i
\log_2 (1 / p_i)\).\footnote{We assume \(0 \log (1 / 0) = 0\).}  He
then proved that, if $P$ is over an alphabet, then we can assign each
character with probability \(p_i > 0\) a prefix-free binary codeword
of length \(\lceil \log_2 (1 / p_i) \rceil\), so the expected
codeword length is less than \(H (P) + 1\);  we cannot, however,  assign
them codewords with expected length less than \(H (P)\).\footnote{In
  fact, these bounds hold for any size of code alphabet; we assume
  throughout that codewords are binary, and by $\log$ we always mean
  $\log_2$.}  Shannon's proof of his upper bound is simple: without
loss of generality, assume \(p_1 \geq \cdots \geq p_n > 0\); for \(1
\leq i \leq n\), let \(b_i = \sum_{j = 1}^{i - 1} p_j\); since \(|b_i
- b_{i'}| \geq p_i\) for \(i' \neq i\), the first \(\lceil \log (1 /
p_i) \rceil\) bits of $b_i$'s binary representation uniquely identify
it; let these bits be the codeword for the $i$th character.  The
codeword lengths do not change if, before applying Shannon's
construction, we replace each $p_i$ by \(1 / 2^{\lceil \log (1 / p_i)
  \rceil}\).  The code then produced is canonical~\cite{SK64}: i.e.,
if a codeword is the $c$th of length $r$, then it is the first $r$
bits of the binary representation of \(\sum_{\ell = 1}^{r - 1} W
(\ell) / 2^\ell + (c - 1) / 2^r\), where $W (\ell)$ is the number of
codewords of length $\ell$.  For example,
\[\begin{array}{rl@{\hspace{5ex}}rl}
1) & 000 & 7) & 1000\\
2) & 001 & 8) & 1001\\
3) & 0100 & 9) & 10100\\
4) & 0101 & 10) & 10101\\
5) & 0110 & \multicolumn{2}{c}{\raisebox{-.5ex}[0ex][0ex]{\vdots}}\\
6) & 0111 & 16) & 11011
\end{array}\]
are the codewords of a canonical code.  Notice the codewords are
always in lexicographic order.

Static prefix coding with a Shannon code stores $s$ in \((H + 1) m + o (m)\)
bits.  An advantage to using a canonical Shannon code is that we can
easily encode each character in \(O (1)\) worst-case time (apart from
first pass) and decode it symmetrically in \(O (\log \log m)\)
worst-case time (see~\cite{MT97}).  To encode a symbol $s[i]$ from $s$, it suffices to know the pair
\( \langle r, c \rangle\) such that the codeword for $s[i]$ is the
$c$-th codeword of length $r$, and the first codeword $l_r$ of length $r$.
Then the codeword for $s[i]$ can be computed in $O(1)$ time
as $l_r+c$.
We store the pair $\langle r,c \rangle$ for the $k$-th symbol
in the alphabet in the $k$-th entry of the array $C$.
The array $L[1..\lceil \log m \rceil]$ contains first codewords of
length $l$ for each $1\leq l \leq \lceil \log m \rceil$.
Thus, if we maintain arrays $L$ and $C$ we can encode a character
from $s$ in $O(1)$ time.

For decoding, we also need a data structure $D$ of size $\log m $
and a matrix $M$ of size $n\times \lceil\log m\rceil$.
For each $l$ such that there is at least one codeword of length
$l$, the data structure $D$ contains the first codeword of
length $l$  padded with $\lceil \log m \rceil - l$
$0$'s. For an integer $q$, $D$ can find the \emph{predecessor}
of $q$ in $D$, $\pred(q,D)=\max\{x\in D|x\leq q\}$.
The entry $M[r,c]$ of the matrix $M$ contains the symbol
$s$, such that the codeword for $s$ is the $c$-th codeword of length
$r$. The decoding algorithm reads the next $\lceil \log m\rceil$
bits into a variable $w$ and finds the predecessor of $w$ in $D$.
When $\pred(w,D)$ is known, we can determine the length $r$ of the
next codeword, and compute its index $c$ as $(w-L[r])\gg (\lceil \log m\rceil -r)$ where $\gg$ denotes the right bit shift operation.

The straightforward binary tree solution allows us to find
predecessors in $O(\log \log m)$ time.
We will see in the next section that predecessor queries on a set
of $\log m$ elements can be answered in $O(1)$ time. Hence, both
encoding and decoding can be performed in $O(1)$ time in  the
static scenario. In the adaptive  scenario, we must find a way
to maintain the arrays $C$ and $L$ efficiently and, in the case
of decoding,  the data structure $D$.

\section{Data structures} \label{sec:data structures}

It is not hard  to speed up the method for
decoding we described in Section~\ref{sec:canonical}.  For example, if
the augmented binary search tree we use as  $D$ is
optimal instead of balanced then, by Jensen's Inequality, we decode
each character in \(O (\log H)\) amortized time.  Even better, if we
use a data structure by Fredman and Willard~\cite{FW93}, then we can
decode each character in \(O (1)\) worst-case time.

\begin{lemma}[Fredman and Willard, 1993] \label{lem:dictionary}
  Given \(O (\log^{1 / 6} m)\) keys, in \(O (\log^{2 / 3} m)\)
  worst-case time we can build a data structure that stores those keys and
  supports predecessor queries in \(O (1)\) worst-case time.
\end{lemma}

\begin{corollary} \label{cor:dictionary}
  Given \(O (\log m)\) keys, in \(O (\log^{3 / 2} m)\) worst-case time
  we can build a data structure that stores those keys and supports
  predecessor queries in \(O (1)\) worst-case time.
\end{corollary}

\begin{proof}
  We store the keys in the leaves of a search tree with degree \(O
  (\log^{1 / 6} m)\), size \(O (\log^{5 / 6} m)\) and height at most
  5.  Each node stores an instance of Fredman and Willard's data structure
  from Lemma~\ref{lem:dictionary}: each  data structure associated  with
  a leaf stores \(O
  (\log^{1 / 6} m)\) keys and each data structure associated with an
  internal node stores
  the first key in each of its children's data structures.  It is
  straightforward to build the search tree in \(O (\log^{2 / 3 + 5 /
    6} m) = O (\log^{3 / 2} m)\) time and implement queries in \(O
  (1)\) time.
\end{proof}

In Lemma~\ref{lem:dictionary} and Corollary~\ref{cor:dictionary} we
assume that multiplication and finding the most significant bit of
an integer can be performed in constant time. As shown in~\cite{ABT99},
we can implement the data structure of Lemma~\ref{lem:dictionary}
using $AC^0$ operations only.
We can restrict the set of elementary operations to bit operations
and table look-ups by increasing the space usage and preprocessing time
 to $O(m^{\eps})$.
In our case  all keys in the data structure $D$ are bounded
by $m$; hence, we can construct in $O(m^{\eps})$ time a look-up table
that uses $O(m^{\eps})$ space and allows us to multiply two integers
or find the most significant bit of an integer, in constant
time.

Corollary~\ref{cor:dictionary} is useful to us because the data structure
it describes not only supports predecessor queries in \(O (1)\)
worst-case time but can also be built in  time polylogarithmic in
$m$; the latter property will let our adaptive Shannon coder keep its
data structures nearly current by regularly rebuilding them.  The array $C$
and matrix $M$ cannot be built in \(o (n)\) time, however, so we
combine them in a data structure that can be updated incrementally.  Arrays $C[]$ and $L[]$, and the matrix $M$ defined in the
previous section, can be rebuilt as described in the next Lemma.

\begin{lemma}\label{lemma:rebuild}
If codeword lengths of $f\geq \log m$ symbols are changed,
we can rebuild arrays $C[]$ and $L[ ]$, and update the matrix $M$ in
$O(f)$ time.
\end{lemma}
\begin{proof}
We maintain an array $S$ of doubly-linked lists. The doubly-linked list
$S[l]$ contains all symbols with codeword length $l$ sorted by
their codeword indices, i.e, the codeword of the $c$-th symbol
in array $S[r]$ is the $c$-th codeword of length $r$.
The number of codewords with lengths $l$ for each $1\leq l\leq \lceil
\log m \rceil$ is stored in the array $W[]$.
If the  codeword length of a symbol $a$ is changed
from $l_1$ to $l_2$, we replace $a$ with $\last[l_1]$ in
$S[l_1]$ and append $a$ at the end of $S[l_2]$.
Then, we set $C[k_l]= C[k_a]$ and $C[k_a]=\langle l_2, W[l_2]\rangle$,
where $k_a$ and $k_l$ are indices of $a$ and $\last[l_1]$ in the
array $C[]$.
Finally, we increment $W[l_2]$, decrement $W[l_1]$ and update
$\last[l_1]$ and $\last[l_2]$.
We also update entries $M[l_1,c_a]$, $M[l_1,W[l_1]-1]$ and $M[l_2,W[l_2] ]$
in the matrix $M$ accordingly, where $c_a$ is the index of $a$'s
codeword  before the update operation.
Thus the array $C$ and the matrix $M$ can be updated in $O(1)$ time
when the codeword length of a symbol is changed.

When codeword lengths of all $f$ symbols are changed, we can compute the
array $L$ from scratch in $O(\log m) =O(f)$ time.
\end{proof}

\section{Algorithm} \label{sec:algorithm}

The main idea of the algorithm of~\cite{KN??}, that achieves $O(1)$
amortized encoding cost per symbol, is \emph{quantization} of
probabilities.
The Shannon code is maintained for the probabilities
$\tilde{p}_j=\frac{\lceil i/q\rceil}{\lfloor\occ{a_i}{s[1..i]}/q\rfloor}$
where $\occ{a_j}{s[1..i]}$ denotes the number of occurrences of the symbol
$a_j$ in the string $s[1..i]$ and the parameter $q= \Theta(\log m)$.
The symbol $a_i$ must occur $q$ times
before  the denominator of the fraction $\tp_i$ is incremented by 1.
Roughly speaking, the value of $\tp_i$, and hence the codeword length
of $a_i$,  changes at most once after $\log m$ occurrences of $a_i$.
As shown in Lemma~\ref{lemma:rebuild},
we can rebuild the arrays $C[]$ and $L[]$ in $O(\log m)$ time.
Therefore encoding can be implemented in $O(1)$ amortized time.
However, it is not clear how to use this approach to obtain constant
worst-case time per symbol.


In this paper a different approach is used. Symbols $s[i+1],
s[i+2],\ldots, s[i+d]$ are
encoded with a Shannon code for the
prefix $s[1]s[2]\ldots s[i-d]$ of the input string.
Recall that in a traditional adaptive code the symbol $s[i+1]$ is
encoded with a code for $s[1]\ldots s[i]$.
While symbols $s[i+1]\ldots s[i+d]$ are encoded, we build an optimal
code for $s[1]\ldots s[i]$. The next group of symbols,
i.e. $s[i+d+1]\ldots s[i+2d]$ will be encoded with a Shannon
code for $s[1]\ldots s[i]$,  and the code for  $s[1]\ldots s[i+d]$
will be simultaneously rebuilt in the background.
Thus every symbol $s[j]$ is encoded with a Shannon code
for the prefix $s[1]\ldots s[j-t]$, $d\leq t < 2d$, of the input string.
That is, when a symbol $s[i]$ is encoded, its codeword length equals
\[\left\lceil\log \frac{i+n-t}{\max\left(\occ{s[i]}{s[1..i-t]},1\right) }\right\rceil\,.\]
We increased the enumerator of the fraction by $n$ and the denominator
is always at least $1$ because
we assume that  every character is assigned a codeword of length
$\lceil\log n -d \rceil$ before  encoding starts. We make this
 assumption only to  simplify the description of our algorithm.
There are others methods of dealing with characters that occur for
the first time in the input string that  are more practically efficient,
see e.g.~\cite{Knu85}. The method of~\cite{Knu85} can also be
used in our algorithm, but
it would not change the total encoding length.

Later we will show that the delay of at most $2d$ increases the length of
encoding only by a lower order term.
Now we turn to the description of the  procedure that updates
the code,
i.e. we will show how the code for $s[1]\ldots s[i]$ can
be obtained from the code for $s[1]\ldots s[i-d]$.

Let $\calC$ be an optimal code for $s[1]\ldots s[i-d]$ and
$\calC'$ be an optimal code for $s[1]\ldots s[i]$. As shown in
section~\ref{sec:canonical}, updating the code is equivalent
to updating the arrays $C[]$ and $L[]$,  the matrix $M$, and the
data structure $D$.
Since a group of $d$ symbols contains at most $d$ different symbols,
we must change codeword lengths of at most $d$ codewords.
The list of symbols $a_1,\ldots,a_k$ such that the codeword
length of $a_k$ must be changed can be constructed in $O(d)$
time. We can construct  an array  $L[]$ for the code $\calC'$
in $O(\max(d,\log m))$ time by Lemma~\ref{lemma:rebuild}.
The matrix $M$ and the array $C[]$ can be  updated in $O(d)$ time because
only $O(d)$ cells
of $M$ are modified. However, we cannot build
 new versions of $M$ and $C[]$ because they  contain $\Theta(n\log m)$
and $\Theta(n)$ cells respectively.
Since we must obtain the new version of $M$ while the old version
is still used, we modify $M$ so that each  cell of $M$ is allowed
to contain two different values, an old one and a new one.
For each cell $(r,c)$ of $M$ we store two values $M[r,c].old$ and
$M[r,c].new$ and the separating value $M[r,c].b$: when the symbol
$s[t]$, $t<M[r,c].b$, is decoded, we use $M[r,c].old$;  when the symbol
$s[t]$, $t\geq M[r,c].b$, is decoded, we use $M[r,c].new$.
The procedure for updating $M$ works as follows: we visit all cells
of $M$ that were modified when the code $\calC$ was constructed.
For every such cell we set $M[r,c].old=M[r,c].new$ and $M[r,c].b=+\infty$.
 Then, we add the new values for those cells of $M$ that must be
modified. For  every  cell that must be modified, the new value is stored in
$M[r,c].new$ and $M[r,c].b$ is set to $i+d$.
The array $C[]$ can be updated in the same way.
When the array $L[]$ is constructed, we can construct the
data structure $D$ in $O(\log^{3/2} m)$ time.

The algorithm described above updates the code if the codeword lengths
of some of the symbols $s[i-d+1]\ldots s[i]$ are changed.
But if some symbol $a$ does not occur in the substring
$s[i-d+1]\ldots s[i]$, its codeword length might still change
in the case when $\log(i)>\log(i_a)$ where $i_a=\max\{j<i|s[j]=a\}$.
We can, however,  maintain the following invariant on the codeword
 length $l_a$:
\begin{equation} \left\lceil \log \frac{i+n}{\max(\occ{a}{s[1..i-2d]},1) } \right\rceil
\leq l_a \leq
\left\lceil \log \frac{i+2n}{\max(\occ{a}{s[1..i-2d]},1) } \right\rceil\,.
\label{eqn:length}
\end{equation}
When the codeword length of a symbol $a$ must be modified, we set its
length to $\lceil \log \frac{i+2n}{\max(\occ{a}{s[1..i]},1) } \rceil$.
All symbols $a$ are also stored in the queue $Q$.
When the code $\calC'$ is constructed, we extract the
first $d$ symbols from $Q$, check whether their codeword lengths must be
 changed, and append those symbols at the end of $Q$.
Thus the codeword length of each symbol is checked at least once
when an arbitrary substring $s[u]\ldots s[u+n]$ of the input string
$s$ is processed.
Clearly, the invariant~\ref{eqn:length}  is maintained.

Thus the procedure for obtaining  the code $\calC'$ from the code $\calC$
consists of the following steps:
\begin{enumerate*}
\item
check symbols $s[i-d+1]\ldots s[i]$ and the first $d$ symbols in the
queue $Q$;
construct a list of codewords whose lengths must be changed;
remove the first $d$ symbols from $Q$ and append them at the end of $Q$
\item
traverse the list $\cN$ of modified cells in the matrix $M$
and the array $C[]$, and remove the old
values from those cells; empty the list $\cN$
\item
update the  array $C[]$  for code $\calC'$;
simultaneously, update the matrix $M$ and construct the list $\cN$ of modified
cells in $M$ and $C[]$
\item
construct the array $L$ and the data structure $D$ for the new code $\calC'$
\end{enumerate*}
Each of the steps described above, except the last one,
can be performed in $O(d)$ time;
the last step can be executed in $O(\max(d,\log^{3/2} m))$ time.
For  $d=\lfloor\log^{3/2}m\rfloor/2$, code $\calC$ can be constructed in
$O(d)$ time.
If the cost of constructing $\calC'$ is evenly distributed among symbols
$s[i-d],\ldots, s[i]$, then we spend $O(1)$ extra time when each  symbol
$s[j]$ is processed.
Since $\occ{a}{s[1..i-\lfloor\log^{3/2}m\rfloor]}\geq \max(\occ{a}{s[1..i]}-\lfloor \log^{3/2} m\rfloor,1)$,
we need at most
$$\left\lceil\log \frac{i+2n}{\max\left(\occ{s[i]}{s[1..i]}-\lfloor\log^{3/2}m \rfloor,1\right) } \right\rceil $$
bits to encode $s[i]$.
\begin{lemma} \label{lem:algorithm}
  We can keep an adaptive Shannon code such that, for \(1 \leq i \leq
  m\), the codeword for \(s [i]\) has length at most
\[\left\lceil \log \frac{i + 2 n}
  {\max \left( \occ{s [i]}{s [1..i]} - \lfloor \log^{3 / 2} m \rfloor,
      1 \right)} \right\rceil\] and we use \(O (1)\) worst-case time
to encode and decode each character.
\end{lemma}

Gagie~\cite{Gag07} and Karpinski and Nekrich~\cite{KN??} proved
inequalities that, together with Lemma~\ref{lem:algorithm}, immediately yield our result.
\begin{theorem} \label{thm:main}
  We can encode $s$ in at most \((H + 1) m + o (m)\) bits with an
  adaptive prefix coding algorithm that encodes and decodes each
  character in \(O (1)\) worst-case time.
\end{theorem}
\noindent For the sake of completeness, we summarize and prove their inequalities as the following lemma:
\begin{lemma} \label{lem:analysis}
  \(\displaystyle \raisebox{-8ex}{} \sum_{i = 1}^m \left\lceil \log
    \frac{i + 2 n} {\max \left( \occ{s [i]}{s [1..i]} - \lfloor
        \log^{3 / 2} m \rfloor, 1 \right)} \right\rceil \leq (H + 1) m
  + O (n\log^{5/2} m)\,.\)
\end{lemma}

\begin{proof}
  Let
\begin{eqnarray*}
L & = & \sum_{i = 1}^m \left\lceil \log \frac{i + 2 n}
    {\max \left( \occ{s [i]}{s [1..i]} - \lfloor \log^{3 / 2} m \rfloor, 1 \right)}
    \right\rceil\\
& < & \sum_{i = 1}^m \log (i + 2 n) -
    \sum_{i = 1}^m \log \max \left( \occ{s [i]}{s [1..i]} - \lfloor \log^{3 / 2} m \rfloor, 1 \right) + m\,.
\end{eqnarray*}
Since \(\left\{ \rule{0ex}{2ex} \occ{s [i]}{s [1..i]}\,:\,1 \leq i
  \leq m \right\}\) and \(\left\{ \rule{0ex}{2ex} j\,:\,1 \leq j \leq
  \occ{a}{s},\,\mbox{$a$ a character} \right\}\) are the same
multiset,
\begin{eqnarray*}
L & < & \sum_{i = 1}^m \log (i + 2 n) -
    \sum_a \sum_{j = 1}^{\occ{a}{s} - \lfloor \log^{3 / 2} m \rfloor} \log j + m\\
& \leq & \sum_{i = 1}^m \log i + 2 n \log (m + 2 n) -
    \sum_a \sum_{j = 1}^{\occ{a}{s}} \log j + n \log^{5 / 2} m + m\\
& = & \log (m!) - \sum_a \log (\occ{a}{s}!) + m + O (n \log^{5 / 2} m)\,.
\end{eqnarray*}
Therefore, by Stirling's Formula,\footnote{Since \(\log (m!) - \sum_a
  \log (\occ{a}{s}!) = \log \left( m! / \prod_a \occ{a}{s}! \right)\)
  is the logarithm of the number of ways to arrange the characters in
  $s$, from this point we could also establish our claim by purely
  information theoretic arguments.}
\begin{eqnarray*}
L & \leq & m \log m - m \ln 2 - \sum_a \left( \rule{0ex}{2ex}
    \occ{a}{s} \log \occ{a}{s} - \occ{a}{s} \ln 2 \right) +
    m + O (n \log^{5 / 2} m)\\
& = & \sum_a \occ{a}{s} \log \frac{m}{\occ{a}{s}} + m + O (n \log^{5 / 2} m)\\
& = & (H + 1) m + O (n \log^{5 / 2} m)\,.
\end{eqnarray*}
The second line of the above equality uses the fact that $m=\sum \occ(a)$.
\end{proof}

\section{Lower bound} \label{sec:lower bound}

It is not difficult to show that any prefix coder uses at least \((H + 1) m
- o (m)\) bits in the worst case (e.g., when $s$ consists of \(m - 1\)
copies of one character and 1 copy of another, so \(H m < \log m + \log e\)).  However, this does
not rule out the possibility of an algorithm that always uses, say, at
most \(m / 2\) more bits than static Huffman coding.
Vitter~\cite{Vit87} proved such a bound is unachievable with an
adaptive Huffman coder, and we now extend his result to all adaptive
prefix coders.  This implies that for an adaptive prefix coder to have
a stronger worst-case upper bound than ours (except for lower-order
terms), that bound can be in terms of neither the empirical entropy
nor the number of bits used by static Huffman coding.\footnote{Notice we do not
exclude the possibility of natural probabilistic settings in
which our algorithm is suboptimal --- e.g., if $s$ is drawn from a
memoryless source for which a Huffman code has smaller redundancy than
a Shannon code, then adaptive Huffman coding almost certainly achieves
better asymptotic compression than adaptive Shannon coding --- but in
this paper we are interested only in worst-case bounds.}

\begin{theorem} \label{thm:lower bound}
  Any adaptive prefix coder stores $s$ in at least \(m - o (m)\) more
  bits in the worst case than static Huffman coding.
\end{theorem}

\begin{proof}
  Suppose \(n = m^{1 / 2} = 2^\ell + 1\) and the first $n$ characters
  of $s$ are an enumeration of the alphabet.  For \(n < i \leq m\),
  when the adaptive prefix coder reaches \(s [i]\), there are at least
  two characters assigned codewords of length at least \(\ell + 1\);
  therefore, in the worst case, the coder uses at least \((\ell + 1) m
  - o (m)\) bits.  On the other hand, a static prefix coder can assign
  codewords of length $\ell$ to the \(n - 2\) most frequent characters
  and codewords of length \(\ell + 1\) to the two least frequent ones,
  and thus use at most \(\ell m + o (m)\) bits.  Therefore, since a Huffman code minimizes the
expected codeword length, any adaptive prefix coder uses at least \(m
- o (m)\) more bits in the worst case than static Huffman coding.
\end{proof}

\section{Online stable sorting} \label{sec:sorting}

Consider $s$ as a multiset of characters and suppose we want to sort
it stably, online and using only binary comparisons.  A stable sort is
one that preserves the order of equal elements, and by `online' we
mean every comparison must have the character we read most recently as
one of its two arguments.  Gagie~\cite{Gag07} noted that, by replacing
Shannon's construction by a modified construction due to Gilbert and
Moore~\cite{GM59}, his coder can be used to sort $s$ using \((H + 2) m
+ o (m)\) comparisons and \(O (\log n)\) worst-case time for each
comparison.  
We can use
our results to speed up Gagie's sorter when \(n = o (\sqrt{m / \log
  m})\).

Whereas Shannon's construction assigns a prefix-free binary codeword
of length \(\lceil \log (1 / p_i) \rceil\) to each character with
probability \(p_i > 0\), Gilbert and Moore's construction assigns a
codeword of length \(\lceil \log (1 / p_i) \rceil + 1\).  If we take
the trie of the codewords, label the leaves from left to right with
the characters in the alphabet and label each internal node with the
label of the rightmost leaf in its left subtree, the result is a
leaf-oriented binary search tree.  Building this tree takes \(O (n)\)
time because, unlike Shannon's construction, we do not need to
sort the characters by probability.  Hence, although we don't know how
to update the alphabetic tree efficiently, we can construct it from scratch
in $O(n)$ time. We can apply the same approach as in previous sections,
and use searching  with delays: while we use the optimal alphabetic tree for
$s[1..i-n/2]$ to identify symbols $s[i],s[i+1],\ldots, s[i+n/2]$,
we construct the tree for $s[1..i]$ in the background. If we use
$O(1)$ time per symbol to construct the next optimal tree, the next
tree will be completed when $s[i+n/2]$ is identified.
Hence, we identify each character using at most
\[\left\lceil \log \frac{i + n}
  {\max \left( \occ{s [i]}{s [1..i]} - n, 1 \right)} \right\rceil +
1\] comparisons and \(O (1)\) worst-case time for each comparison.
A detailed description of the algorithm wil be given in the full version of
this paper.
The following technical lemma, whose proof we omit because it is
essentially the same as that of Lemma~\ref{lem:analysis}, bounds the
total number of comparisons we use to sort $s$ and, thus, implies our
speed-up.

\begin{lemma} \label{lem:sorting analysis}
\[\sum_{i = 1}^m \left\lceil \log \frac{i + n}
  {\max \left( \occ{s [i]}{s [1..i]} - n, 1 \right)} \right\rceil + m
\leq (H + 2) m + O (n^2 \log m)\,.\]
\end{lemma}

\begin{theorem} \label{thm:sorting}
  If \(n = o (\sqrt{m / \log m})\), then we can sort $s$ stably and
  online using \((H + 2) m + o (m)\) binary comparisons and \(O (1)\)
  worst-case time for each comparison.
\end{theorem}

We prove the following theorem by essentially the same arguments as
for Theorem~\ref{thm:lower bound}.  It shows that, when \(n = o
(\sqrt{m / \log m})\), our sorter is essentially optimal.  We leave as
an open problem finding a sorter that uses \((H + 2) m + o (m)\)
comparisons and \(O (1)\) worst-case time per comparison when $n$ is
closer to $m$.

\begin{theorem} \label{thm:sorting LB}
  Any online stable sort uses at least \((H + 2) m - o (m)\) binary
  comparisons in the worst case.
\end{theorem}

\proof Suppose \(n = m^{1 / 2} = 2^{\ell + 1} + 2\), $s$ contains only
the even-numbered characters in the alphabet, and the first \(2^\ell +
1\) characters of $s$ are an enumeration of those characters.  It is
not difficult to show that, for \(2^\ell + 1 < i \leq m\), the online
stable sorter must determine the identity of \(s [i]\), in case it is
an odd-numbered character; moreover, it must do this before reading
\(s [i + 1]\), in case all the remaining characters in $s$ are smaller
than \(s [i]\)'s predecessor in \(\{s [j]\,:\,1 \leq j < i\}\) or
larger than its successor.  Therefore, after \(s [2^\ell + 1]\), we
can view the sorter as processing each character using a binary
decision tree with $n$ leaves, labelled from left to right with the
characters in the alphabet in lexicographic order.  (Since we are
proving a worst-case lower bound, we assume without loss of generality
that the sorter is deterministic.)  For \(2^\ell + 1 < i \leq m\),
when the sorter reaches \(s [i]\), there must be a leaf at depth at
least \(\ell + 2\) that is labelled with an even-numbered character in
the alphabet.  Therefore, even when $s$ contains only the
even-numbered characters in the alphabet, the sorter uses at least
\((\ell + 2) m - o (m)\) comparisons in the worst case.  Since $s$
contains only \(2^\ell + 1\) distinct characters, however, \(H \leq
\log (2^\ell + 1)\) and calculation shows \((H + 2) m - o (m) \leq
(\ell + 2) m - o (m)\).  \qed

\section{Other Coding Problems}\label{sec:other}
Several variants of the prefix coding problem were considered and
extensively studied. In the alphabetic coding problem~\cite{GM59},
codewords must be sorted lexicographically, i.e. $i< j\Rightarrow $
$c(a_i)< c(a_j)$, where $c(a_k)$ denotes the codeword of $a_k$.  In
the length-limited coding problem, the maximum codeword length is
limited by a parameter $F >\log n$.  In the coding with unequal letter
costs problem, one symbol in the code alphabet  costs more than
 another and we
 want to minimize the average cost of a codeword.  All of the above
problems were studied in the static scenario.  Adaptive prefix coding
algorithms for those problems were considered in~\cite{Gag04}. In this
section we show that the good upper bounds on the length of the
encoding can be achieved by algorithms that encode in $O(1)$
worst-case time.  The main idea of our improvements is that we encode
a symbol $s[i]$ in the input string with a code that was constructed
for the prefix $s[1..i-d]$ of the input string, where the parameter
$d$ is chosen in such a way that a corresponding (almost) optimal code
can be constructed in $O(d)$ time.  Using the same arguments as in the
proof of Lemma~\ref{lem:analysis} we can show that encoding with
delays increases the length of encoding by an additive term of
$O(d\cdot n \log m)$ (the analysis is more complicated in the case of
coding with unequal letter costs).  We will provide proofs in
the full version of this paper.

{\bf Alphabetic Coding.}  The algorithm of section~\ref{sec:sorting}
can be used for adaptive alphabetic coding. The length of encoding
is $((H + 2) m + o (m)$ provided that $n=o(\sqrt{m/\log m})$.
Unfortunately we cannot use the encoding and decoding methods of
section~\ref{sec:canonical} because the alphabetic coding is not
canonical.  When an alphabetic code for the following group of $O(n)$
symbols is constructed, we also create in $O(n)$ time a table that
stores the codeword of each symbol $a_i$. Such a table can be created
from the alphabetic tree in $O(n)$ time; hence, the complexity of the
encoding algorithm is not increased.  We can decode the next codeword
by searching in the data structure that contains all codewords. Using
a data structure due to Andersson and Thorup~\cite{AT07} we can decode
in $O(\min(\sqrt{\log n},\log \log m)$ time per symbol.
\begin{theorem}
\label{theor:alphabetic}
There is an algorithm for adaptive alphabetic prefix coding that
encodes and decodes each symbol of a string $s$ in $O(1)$ and
$O(\min(\sqrt{\log n}, \log \log m))$ time respectively.  If
$n=o(\sqrt{m/\log m})$, the encoding length is $((H + 2) m + o (m)$.
\end{theorem}

{\bf Coding with unequal letter costs.}  Krause~\cite{Kra62} showed
how to modify Shannon's construction for the case in which code
letters have different costs, e.g., the different durations of dots
and dashes in Morse code.  Consider a binary channel and suppose
\cost{0}\ and \cost{1}\ are constants with \(0 < \cost{0} \leq
\cost{1}\).  Krause's construction gives a code such that, if a symbol
has probability $p$, then its codeword has cost less than \(\ln (p) /
C + \cost{1}\), where the channel capacity $C$ is the largest real
root of \(e^{- \cost{0} \cdot x} + e^{- \cost{1} \cdot x} = 1\) and
$e$ is the base of the natural logarithm.  It follows that the
expected codeword cost in the resulting code is \(H \ln 2 / C +
\cost{1}\), compared to Shannon's bound of \(H \ln 2 / C\).  Based on
Krause's construction, Gagie gave an algorithm that produces an
encoding of $s$ with total cost at most \(\left( \frac{H \ln 2}{C} +
  \cost{1} \right) m + o (m)\) in \(O (m \log n)\) time.
Since the code of Krause~\cite{Kra62} can be constructed in $O(n)$ time,
we  can use the encoding with delay $n$ and achieve $O(1)$ worst-case
time.
Since the
costs are constant and the minimum probability is \(\Omega (1 / m)\),
the maximum codeword length is \(O (\log m)\).  Therefore, we can decode
using the data structure described above.
\begin{theorem}
  There is an algorithm for adaptive prefix coding with unequal letter
  costs that encodes and decodes each symbol of a string $s$ in $O(1)$
  and $O(\min(\sqrt{\log n}, \log \log m))$ time respectively.  If
  $n=o(\sqrt{m/\log m})$, the encoding length is $\left( \frac{H \ln
      2}{C} + \cost{1} \right) m + o(m)$.
\end{theorem}

{\bf Length-limited coding.}  Finally, we can design an algorithm for
adaptive length-limited prefix coding by modifying the algorithm of
section~\ref{sec:algorithm}.  Using
the same method as in~\cite{Gag04} --- i.e., smoothing the
distribution by replacing each probability with a weighted average of
itself and \(1 / n\) --- we set the codeword length of symbol $s[i]$
to $\lceil \log \frac{2^f}{(2^f-1)x+1/n} \rceil$ instead of $\lceil
\log\frac{1}{x}\rceil$, where $x= \frac{\max \left( \occ{s [i]}{s
      [1..i]} - \lfloor \log^{3 / 2} m \rfloor, 1 \right)}{i + 2 n}$
and $f=F-\log n$.  We observe that the codeword lengths $l_i$ of this
modified code satisfy the Kraft-McMillan inequality:
\begin{eqnarray*}
\sum_i 2^{-l_i}
& \leq & \sum_x ((2^f - 1) x + 1 / n) / 2^f\\
& = & \sum_x (2^f x) / 2^f - \sum_x  x / 2^f + n (1 / n) / 2^f=1\\
\end{eqnarray*}
Therefore we can construct and maintain a canonical prefix code with
codeword lengths $l_i$.  Since $\frac{2^f}{(2^f-1)x+1/n} \leq
\min(\frac{2^f}{(2^f-1)x},\frac{2^f}{1/n})$,
$\lceil \log \frac{2^f}{(2^f-1)x+1/n}\rceil \leq
\min( \lceil \log \frac{2^f}{(2^f-1)x}\rceil, \log n + f)\,.$. Thus
the codeword length is always smaller than $F$.  We can estimate the
encoding length by bounding the first part of the above expression:
$\lceil \log \frac{2^f}{(2^f-1)x}\rceil < x +1 + \log\frac{2^f
  +1}{2^f}$ and $\log\frac{2^f +1}{2^f} = (1/2^f)\log
(1+\frac{1}{2^f})^{2^f}\leq \frac{1}{2^f\ln 2}$.  Summing up by all
symbols $s[i]$, the total encoding length does not exceed
\[\sum_{i = 1}^m  \log
\frac{i + 2 n} {\max \left( \occ{s [i]}{s [1..i]} - \lfloor \log^{3 /
      2} m \rfloor, 1 \right)} + m + \frac{m}{2^f\ln 2}\,.\] We can
estimate the first term in the same way as in
Lemma~\ref{lem:analysis}; hence, the length of the encoding is
$(H+1+\frac{1}{2^f\ln 2})m + O(n\log^{5/2} m)$. We thus obtain the following
theorem:
\begin{theorem}
  There is an algorithm for adaptive length-limited prefix coding that
  encodes and decodes each symbol of a string $s$ in $O(1)$ time.  The
  encoding length is $(H+1+\frac{1}{2^f\ln 2})m + O(n\log^{5/2} m)$,
  where $f = F-\log n$ and $F$ is the maximum codeword length.
\end{theorem}


\bibliographystyle{plain} \bibliography{prefix}

\end{document}